\documentclass[10pt, aps, pra, twocolumn]{revtex4-1}
\pdfoutput=1
\usepackage[utf8]{inputenc}
\usepackage[english]{babel}
\usepackage[T1]{fontenc}
\usepackage{amsmath, mathtools}
\usepackage{amsfonts}
\usepackage{amssymb}
\usepackage{makeidx}
\usepackage{graphicx}
\usepackage{amsthm}
\usepackage{color}
\usepackage{multirow}
\usepackage[sort&compress]{natbib}
\newtheorem{theorem}{Theorem}
\newtheorem{prop}{Proposition}
\newtheorem*{prop*}{Proposition}
\newtheorem{lemma}{Lemma}
\newtheorem{coro}{Corollary}
\theoremstyle{remark}
\newtheorem{exm}{Example}
\newtheorem{remark}{Remark}
\theoremstyle{definition}

\DeclareMathOperator{\states}{\mathfrak{S}}

\def\Ha{\mathcal{H}}

\def\dH{\dim({\mathcal{H}})}
\def\Ce{\mathcal{C}}

\DeclareMathOperator{\conv}{conv}

\DeclareMathOperator{\lin}{span}

\DeclareMathOperator{\degcom}{DegCom}
\DeclareMathOperator{\Tr}{Tr}
\DeclareMathOperator{\id}{id}
\usepackage{dsfont}
\def\I{\mathds{1}}
\def\tmin{\dot{\otimes}}
\def\tmax{\hat{\otimes}}
\def\treal{\tilde{\otimes}}
\def\<{\langle}
\def\>{\rangle}
\def\chsh{\text{CHSH}}
\def\diag{\text{diag}}
\def\off{\text{off}}
\def\inn{\text{in}}
\def\out{\text{out}}
\DeclareMathOperator{\swap}{SWAP}
\DeclareMathOperator{\Choi}{Choi}
\newcommand{\A}{\mathsf{A}}
\newcommand{\B}{\mathsf{B}}
\newcommand{\C}{\mathsf{C}}
\newcommand{\D}{\mathsf{D}}
\newcommand{\E}{\mathsf{E}}
\newcommand{\F}{\mathsf{F}}
\newcommand{\G}{\mathsf{G}}
\newcommand{\M}{\mathsf{M}}
\newcommand{\T}{\mathsf{T}}
\begin{document}
\title{Structure of quantum and classical implementations of Popescu-Rohrlich box}

\author{Anna Jen\v{c}ov\'{a}}
\affiliation{Mathematical Institute, Slovak Academy of Sciences, \v Stef\' anikova 49, Bratislava, Slovakia}

\author{Martin Pl\'{a}vala}
\email{martin.plavala@mat.savba.sk}
\affiliation{Mathematical Institute, Slovak Academy of Sciences, \v Stef\' anikova 49, Bratislava, Slovakia}
\affiliation{Naturwissenschaftlich-Technische Fakult\"{a}t, Universit\"{a}t Siegen, 57068 Siegen, Germany}

\begin{abstract}
We construct implementations of the PR-box using quantum and classical channels as state spaces. In both cases our constructions are very similar and they share the same idea taken from general probabilistic theories and the square state space model. We construct all quantum qubit channels that maximally violate a given CHSH inequality, we show that they all are entanglement-breaking channels, that they have certain block-diagonal structure and we present some examples of such channels.
\end{abstract}

\maketitle

\section{Introduction}
The Bell non-locality is a well known topic in quantum theory, yet we still lack full understanding of its implications. The research of Bell non-locality was inspired by the famous paradox of Einstein, Podolsky and Rosen \cite{EinsteinPodolskyRosen-paradox} that questioned the completeness of quantum theory.

As it was later shown by Bell \citep{Bell-ineq} the EPR paradox does not question the completeness of quantum theory, but it rather separates it from any other classical theory. This was demonstrated by the well-known Bell inequalities, that constrain any classical theory but are violated by quantum theory. Probably the most well known and most studied of Bell inequalities is the CHSH inequality \cite{ClauserHorneShimonyHolt-CHSH} that can be violated by quantum states and measurements, but this violation is constrained by the Tsirelson bound \cite{Cirelson-bound}. On one hand this shows that quantum theory is more non-local than any classical theory, on the other hand, Popescu and Rohrlich proved that the CHSH inequality may be violated even more by non-signaling correlations \cite{PopescuRohrlich-PRbox, Rohrlich-PRbox}. Since then, it is a long-standing question whether the Tsirelson bound does have any operational meaning for quantum theory.

It was later shown by Barrett that the conditional probability distribution identified in \cite{PopescuRohrlich-PRbox}, now often called the PR-box, can be generated by a non-signaling theory \cite{Barrett-GPTinformation}, most commonly known as the Boxworld GPT.  A realization of the PR-box in the real world would have several rather interesting implications, as shown in \cite{QuekShor-superquantChan, vanDam-communication, BuhrmanChristandlUngerWehnerWinter-crypt}, see also \cite{BrunnerCavalcantiPironioScaraniWehner-review} for a review.

It has been demonstrated several times that non-signaling classical and quantum channels provide realizations of the PR box \cite{BeckmanGottesmanNielsenPreskill-channels, HobanSainz-channels, PlavalaZiman-PRbox, Crepeau-CHSH}. As described in \cite{PlavalaZiman-PRbox}, such constructions can be put into the framework of general probabilistic theories (or GPT for short), in the following way: since the set of classical (or quantum) channels is compact and convex, it can be seen as a state space of some GPT. One can see non-signaling channels as elements of a joint state space, equivalent to entangled states in quantum theory, and one can describe measurement procedures and CHSH experiments within this theory, yielding maximal violation of the CHSH inequality.

One has to be careful when using channels in this way. The formalism implies that we are able to use the channel only once, hence one can exploit the input state incompatibility of measurements on channels \cite{SedlakReitznerChiribellaZiman-compatibility}. From a realistic view-point it is important to remember that using a non-local channel may take some time and resources, but the channel is considered to permit communication only if the channel is signaling.

In the present paper, we focus on the structure of the non-signaling channels such that the CHSH inequality is maximally violated for some choice of channel measurements. Such channels will be called the PR-channels. All the PR-channels obtained so far are in fact classical-to-classical and it is natural to ask whether there are some truly quantum non-signaling PR-channels. Another important question is the possibility of instantaneous implementation of such channels. Since PR-channels are non-signaling, their instantaneous implementation is not forbidden by special theory of relativity simply because no information is transferred, yet it is believed that such implementations do not exist.

We make a step toward addressing these questions, in that we present a characterization of the structure of all implementations of the PR-box in the framework of GPTs, especially for theories in which classical and quantum channels play the role of states. In any GPT, the pairs of measurements appearing in such implementations must be maximally incompatible and we show how the corresponding bipartite states are constructed from such pairs.

We apply the obtained results mainly to qubit bipartite non-signaling channels, where both parts of the input and output are qubit spaces. Here we give a full description of all possible pairs of maximally incompatible two-outcome channel measurements and of all qubit PR-channels. In particular, we prove that all these channels are necessarily entanglement-breaking. We believe that our results will bring more insight into the structure of PR-channels, in particular to the question of existence of their instantaneous implementation.

The article is organized as follows: in Sec. \ref{sec:GPT} we give a brief overview of general probabilistic theories as it will be used in later calculations. In Sec. \ref{sec:finding} we present the method of finding all bipartite non-signaling states that maximally violate the CHSH inequality and show that maximally incompatible measurements are necessary, the main result is stated as Thm. \ref{thm:finding-result}.  In Sec. \ref{sec:cPR}, we present the (known) PR-box implementation by classical channels in the light of  the results of Sec. \ref{sec:finding}. The main purpose of Section \ref{sec:qPR} is to introduce the GPT of quantum channels and show how the known PR-box implementations (and their slight generalizations) are obtained by our construction. In Sec. \ref{sec:qubit} we derive the structure of all qubit PR-channels. The main result is in Thm. \ref{thm:qubit-ETBchannel} where we prove that all qubit PR-channels must be entanglement-breaking and we provide some examples. Some more technical proofs can be found in the appendices.

\section{Overview of general probabilistic theories and tensor products} \label{sec:GPT}

General probabilistic theories (GPTs for short) provide a framework that uses operational axioms to describe various
possible physical theories. GPTs include the classical and the quantum theory, hence this setting allows us to compare
these two theories as well as to construct theories that are different from both. Among other things, GPTs provide a
framework to describe measurements of a physical systems in a general and mathematically clear way. We briefly introduce
the formalism below, which allows us to obtain general results applicable to various state spaces. Since it will be
sufficient for all of our calculations, we will only consider theories with finite dimensional state spaces. This
section is only intended to settle the notation that we use; for a full review of GPTs we refer the reader to
\cite{JanottaHinrichsen-GPTreview}. A nice introduction to GPTs, including a historical account, can be found in
\cite{lami2017nonclassical}.

The main idea behind GPTs is the following: both in (finite-dimensional) classical and quantum theory, the state space (i.e. the set of all preparation procedures) is a compact convex subset of a real finite-dimensional vector space. We will generalize both theories by assuming that the state space is some compact convex subset of a finite-dimensional vector space.

Any GPT can be described as a set of physical systems it contains. In operational terms, any system is determined by the set of allowed preparation procedures (states) and yes-no experiments (effects).

The set of states has a natural convex structure, because if $x$ and $y$ are preparable states of a system, then we can randomize the preparation to prepare the convex combination $\lambda x + (1-\lambda) y$, which we again postulate to be a viable state. One can also argue that the state space must be closed (in a for now unspecified topology) as if we can prepare a series of states $x_n$ converging to $x$, then we should be also able to prepare $x$ (albeit with infinite resources).

An effect is described by an affine function mapping states to probabilities of the "yes" outcome, represented by the interval $[0, 1]$. Recall that a function $f$ is affine if it respects the convex structure:  $f(\lambda x + (1-\lambda) y ) = \lambda f(x) + (1-\lambda) f(y)$, for any states $x$ and $y$ and any $\lambda\in [0,1]$.  We will require that for any two distinct states there is an effect that distinguished these states strictly better than a random guess. This implies that the state space must be bounded. Indeed, if the state space would have a direction of recession in which it would go to infinity (i.e. it would be unbounded), any effect would have to be zero in this direction, so we would be unable to distinguish these states.

We also require the state space to be embedded in a real finite-dimensional vector space equipped with the Euclidean topology. This last assumption is only practical as it allows us to stick only to rather simple mathematics. Note that it is well known that any bounded and closed subset of a finite-dimensional real vector space is compact, hence the state space is compact.

\subsection{Structure of GPTs}

As we have seen above, the basic mathematical framework for GPT consists of compact convex sets in finite dimensional Euclidean spaces and affine functions on them. Below, we introduce the notations and further assumptions, applied throughout the paper.

Let $V$ be a finite dimensional real vector space with the standard Euclidean topology and let $K \subset V$ be a
compact convex set, interpreted as the state space of a system in a GPT. Let $A(K)$ denote the linear space of real-valued affine functions on $K$. We will denote constant functions by the value they attain. Let $f, g \in A(K)$, then we introduce an ordering to $A(K)$ as follows: $f \geq g$ if and only if for every $x \in K$ we have $f(x) \geq g(x)$. Let $A(K)^+ = \{ f \in A(K): f \geq 0 \}$ denote the convex, closed, generating, pointed cone of positive functions and let $E(K) = \{f \in A(K): 0 \leq f \leq 1 \}$ denote the set of effects on $K$, called the effect algebra. The effect algebra $E(K)$ is important in GPTs since (as will be explained below) the effects describe two-outcome measurements of the theory.

Denote $A(K)^*$ the dual of $A(K)$ and denote $A(K)^{*+}$ the positive cone dual to $A(K)^+$, i.e.
\begin{equation*}
A(K)^{*+} = \{ \psi' \in A(K)^*: \psi'(f) \geq 0, \forall f \in A(K)^+\}.
\end{equation*}
The cone $A(K)^{*+}$ gives rise to an ordering on $A(K)^*$: let $\psi, \varphi \in A(K)^*$, then $\psi \geq \varphi$ if
and only if $\psi - \varphi \in A(K)^{*+}$, i.e. if $\psi - \varphi \geq 0$. The state space $K$ is affinely isomorphic
to the subset $\{ \psi' \in A(K)^{*+}: \psi'(1) = 1 \}$, see \cite[Chapter 1,
Theorem 4.3]{AsimowEllis}. Note that by definition, we consider all states to be normalized, which is expressed by the condition 
$\psi(1) = 1$.  For simplicity we will omit the above isomorphism and  treat $K$ as a subset of $A(K)^*$.

It follows that any $0 \ne \psi\in A(K)^{*+}$ can be expressed uniquely as $\psi=\alpha x$ for $\alpha>0$ and $x\in K$. The cone $A(K)^{*+}$ is generating, so we can express every $\psi \in A(K)^*$ as $\psi = \alpha x - \beta y$ for some $x, y \in K$ and $\alpha, \beta \in \mathbb{R}$, $\alpha, \beta \geq 0$.

Now we will present a simple definition of a two-outcome measurement in GPTs. Generally speaking, a measurement is a procedure that assigns probabilities to possible measurement outcomes. For simplicity we will restrict ourselves to outcomes labeled by the numbers $-1$, $1$. Let $\A$ be such a measurement and let $P_x(\epsilon | \A)$ denote the probability of obtaining the outcome labeled as $\epsilon\in \{-1,1\}$ when we measure a system in the state $x \in K$. Since the outcome probabilities must respect probabilistic mixtures, the map $x\mapsto P_x(1 | \A)$ is an effect and $x\mapsto P_x(-1 | \A) = 1 - P_x(1 | \A)$ is an effect as well. Although in general the measurements in the theory may be restricted, in this work we will assume that any effect $f\in E(K)$ gives rise to such a measurement, determined for $x\in K$ as
\begin{align*}
P_x(1 | \A) &= f(x), \\ P_x(-1 | \A) &= (1 - f)(x) = 1 - f(x).
\end{align*}
Such an assumption is called the No-Restriction hypothesis \cite{ChiribellaDArianoPerinotti-GPTpurification}, see also \cite{FilippovGudderHeinosaariLeppajarvi-restrictions} for a recent treatment.

Measurements with a finite number of outcomes are similarly described by collections of effects $f_i\in E(K)$, such that $\sum_i f_i=1$. For a more general treatment of measurements see e.g. \cite[Section 2.2]{Holevo-QT}.

We now list some examples of state spaces, and relations between them, that are basic for the present work.

\begin{exm}[Classical bit]\label{exm:clbit}
The simplest (nontrivial) example is the 1-dimensional simplex $\states_C$. Let the extreme points (deterministic states) in $\states_C$ be denoted by $s_0,s_1$, then $\states_C=\conv(s_0,s_1)$. Since all affine functionals are fully determined by their values at $s_0,s_1$, we have $A(\states_C)\simeq\mathbb R^2$ and the set of effects is identified with $E(\states_C)\simeq [0,1]^2$. Let $\pi \in E(\states_C)$ be the effect determined by $\pi(s_0)=0$, $\pi(s_1)=1$, then $\pi:\states_C\to [0,1] $ is an affine isomorphism.
\end{exm}

\begin{exm}[Quantum state space]\label{exm:quantum}
Let $\Ha$ be a finite dimensional complex Hilbert space. Let $B_h(\Ha)$ denote the set of self-adjoint operators on $\Ha$ and let $\states_{\Ha} = \{ \rho \in B_h(\Ha): \rho \geq 0, \Tr(\rho) = 1 \}$ denote the set of density operators (or states) on $\Ha$, where $\rho \geq 0$ means that $\rho$ is positive semi-definite and $\Tr(\rho)$ denotes the trace of $\rho$.  Let $\I$ denote the identity operator. Then $A(\states_\Ha)\simeq B_h(\Ha)$ and $E(\states_\Ha)\simeq E(\Ha)$, where $E(\Ha)=\{ 0\le E\le \I, E\in B(\Ha)\}$ is the set of quantum effects. The quantum measurements will be thus described by collections $E_1,\dots, E_n$ of positive operators such that $\sum_i E_i=\I$, such a collection is called a POVM \cite{HeinosaariZiman-MLQT}. We will mostly treat the qubit case, that is $\dH=2$.
\end{exm}

\begin{exm}[Square state space]\label{exm:square}
This state space is also called the gbit and appears as the state space of systems in a theory called GNST introduced in \cite{Barrett-GPTinformation}, now also known as the Boxworld GPT. This theory is one of the most commonly used examples of a GPT other than classical or quantum theories. The square $S$ is a state space with four extreme points $s_{00}, s_{10}, s_{01}, s_{11}$, such that
\begin{equation*}
\dfrac{1}{2} (s_{00} + s_{11} ) = \dfrac{1}{2} ( s_{10} + s_{01} ).
\end{equation*}
Let $\pi_0, \pi_{1} \in E(S)$ be given, for $i, j \in \{0, 1\}$, as
\begin{align*}
&\pi_0(s_{ij}) = i,
&&\pi_{1}(s_{ij}) = j,
\end{align*}
then $A(S) = \lin( \{\pi_0, \pi_{1}, 1 \})$ and $E(S) = \conv (\{ \pi_0, 1-\pi_0, \pi_{1}, 1-\pi_{1}, 0, 1 \})$. Moreover, we can see that $\{s_{00}, s_{10}-s_{00},s_{01}-s_{00}\}$ forms a linear basis of $A(S)^*$, which is dual to $\{\pi_0, \pi_{1}, 1 \}$. Note also that the map
\begin{equation*}
s \mapsto (\pi_0(s),\pi_{1}(s)), \qquad s\in S
\end{equation*}
is an affine isomorphism of $S$ onto $[0,1]^2\simeq \states_C\times \states_C$, its inverse is given by
\begin{equation*}
(\lambda,\mu)\mapsto s_{00}+\lambda(s_{10}-s_{00})+\mu(s_{01}-s_{00}),\qquad \lambda,\mu\in [0,1].
\end{equation*}
\end{exm}

The notations $\pi$, $\pi_0$, $\pi_1$ introduced in the above examples will be kept throughout.

\begin{exm}[Classical bit channels]\label{exm:clchannels}
A channel $\Phi: \states_C \to \states_C$ is defined as an affine map of $\states_C$ into itself; a simple example of a channel is the identity channel $\id: \states_C \to \states_C$. Let $\Ce_C$ be the set of all channels (on the classical bit), then since $\pi\circ \Phi$ is an effect on $\states_C$ for every $\Phi\in \Ce_C$ and $\pi$ is an isomorphism, we see that
\begin{equation*}
\Ce_C\simeq E(\states_C)\simeq [0,1]^2\simeq S.
\end{equation*}
Through this isomorphism, we establish a relation between classical bit channels and the Boxworld GPT.

The effect $\pi_0$ ($\pi_{1}$) on $S$ corresponds to the effect in $E(\Ce_C)$, determined as $\Phi\mapsto \pi(\Phi(s_0))$ ($\Phi\mapsto \pi(\Phi(s_1))$). More generally, each effect in $E(\Ce_C)$ is a convex combination of effects of the form
\begin{equation*}
F_{t_i,f_i}(\Phi):=f_i(\Phi(t_i)),
\end{equation*}
where $t_i\in \states_C$ and $f_i\in E(\states_C)$ (clearly, we may restrict to the extreme points $t_i\in \{s_0,s_1\}$).  The corresponding measurement can be seen as a protocol where we  choose the state $t_i$ with some probability $\lambda_i$, input it into the measured channel $\Phi$ and apply the effect $f_i$ to the output of the channel. For such a measurement $\A$ and a channel $\Phi \in \Ce_C$, we have
\begin{equation*}
P_\Phi (1| \A) = \sum_i \lambda_i P_{\Phi(t_i)} (1 | \M_i),
\end{equation*}
where $\M_i$ is the measurement given by $f_i$, see also Fig. \ref{fig:GPT-channel-measurement}.
\end{exm}

\begin{figure}
\includegraphics[width=.8\linewidth]{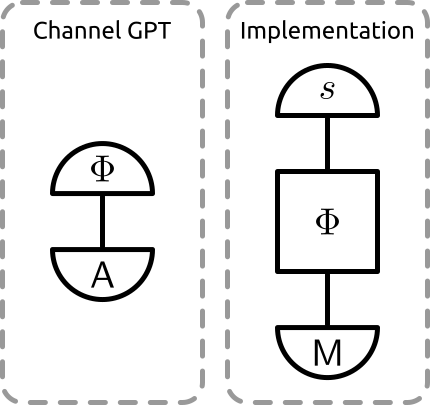}
\caption{Two equivalent ways of describing the same measurement $\A$ on the channel $\Phi \in \Ce_C$. On the left we have the viewpoint of the GPT from example \ref{exm:clchannels} where channels are considered as states and the channel is measured. On the right hand side we have the implementation of the measurement, where $s \in \states_C$ is used as an input into the channel $\Phi$ and the outcome, which is $\Phi(s) \in \states_C$ is measured by $\M$. \label{fig:GPT-channel-measurement}}
\end{figure}

Note that we can obtain similar relations for finite dimensional simplices, their products and channels between them. The (for us) most important example of quantum channels will be described later in Section \ref{sec:qPR}.

\subsection{Incompatible measurements in GPT}

Let $\A$, $\A'$ be two-outcome measurements on $K$ given by the effects $f_\A, f_{\A'} \in E(K)$ respectively, then we say that they are compatible \cite{StevensBusch-CHSHIncomp, Banik-steering, SedlakReitznerChiribellaZiman-compatibility, HeinosaariMiyaderaZiman-compatibility} if there is a four-outcome measurement $\G$ with outcomes $(-1, -1)$, $(-1, 1)$, $(1, -1)$, $(1, 1)$ such that for all $x \in K$ we have
\begin{align*}
P_x(1 | \A) &= P_x((1, 1) | \G) + P_x((1, -1) | \G), \\ P_x(1 | \A') &= P_x((1, 1) | \G) + P_x((-1, 1) | \G).
\end{align*}
Equivalently, measurements $\A$ and $\A'$ are compatible if and only if there is $p \in E(K)$ such that
\begin{align*}
f_\A &\geq p \\ f_{\A'} &\geq p \\ 1 + p &\geq f_\A + f_{\A'}
\end{align*}
see \cite{Plavala-simplex} for a proof. This definition of compatibility of measurements generalizes to GPTs the well-known notion of compatibility (or joint measurability) of POVMs in quantum theory, which itself generalizes the notion of commutativity of projective measures, see \cite{HeinosaariMiyaderaZiman-compatibility} for a review.

The degree of compatibility of $\A$ and $\A'$ is defined as
\begin{align*}
\degcom(\A, \A') = \sup_{\T, \T'} \{ \lambda \in [0, 1]: \lambda \A + (1-\lambda \T), \\
\lambda \A' + (1-\lambda \T') \text{ are compatible} \},
\end{align*}
where $\T$ and $\T'$ are trivial or coin-toss measurements determined by constant effects $\mu 1,\mu'1$ for $\mu,\mu'\in [0,1]$. It is known that we always have $\degcom(\A, \A') \geq \frac{1}{2}$ \cite{BuschHeinosaariSchultzStevens-compatibility} and we say that the measurements $\A$ and $\A'$ are maximally incompatible if $\degcom(\A, \A') = \frac{1}{2}$.

By \cite{JencovaPlavala-maxInc}, the effects $f_\A, f_{\A'} \in E(K)$ correspond to maximally incompatible measurements $\A, \A'$ if and only if there are four points $x_{00}, x_{10}, x_{01}, x_{11} \in K$ satisfying
\begin{equation}
\label{eq:witnessquare1}
\frac{1}{2}( x_{00} + x_{11}) = \frac{1}{2}( x_{10} + x_{01} )
\end{equation}
and such that
\begin{align}
\label{eq:witnessquare2}
f_\A(x_{00}) = f_\A(x_{01}) = f_{\A'}(x_{00}) = f_{\A'}(x_{10}) &= 0, \\ f_\A(x_{10}) = f_\A(x_{11}) = f_{\A'}(x_{01}) = f_{\A'}(x_{11}) &= 1.\label{eq:witnessquare3}
\end{align}
We will call such a set of points $x_{ij}$ a witness square for $f_\A, f_{\A'}$, or equivalently for $\A, \A'$.

\begin{exm}
\label{exm:maxinc_square}
Let $\A_0$ denote the two-outcome measurement on the square determined by effect $\pi_0$ and let $\A_{1}$ be determined by the effect $\pi_{1}$, see Example \ref{exm:square}. It is immediate that $\A_0$ and $\A_{1}$ are maximally incompatible, the witness square being formed by the extreme points $\{s_{ij}\}$. The existence of maximally incompatible measurements on $S$ was first observed in \cite{BuschHeinosaariSchultzStevens-compatibility}.
\end{exm}

We now present an equivalent characterization of maximally incompatible measurements which shows that all such pairs are, in some sense, isomorphic to the pair $\A_0,\A_{1}$ on $S$. This characterization was proved in a more general form in \cite[Corollary 5]{Jencova-incompatible}, we include a short proof for the convenience of the reader.

\begin{prop}
\label{prop:GPT-condMaxInc}
The measurements $\A$, $\A'$ on the state space $K$ corresponding to the effects $f_\A$, $f_{\A'} \in E(K)$ respectively are maximally incompatible if and only if there are affine maps $\iota: S \to K$ and $\Pi: K \to S$ such that $\Pi \circ \iota = \id$ and for $i, j \in \{0, 1 \}$ we have
\begin{align}
f_\A (\iota(s_{ij})) &= \pi_0 (s_{ij}), \label{eq:GPT-comdMaxInc-1} \\ f_{\A'} (\iota(s_{ij})) &= \pi_{1} (s_{ij}). \label{eq:GPT-comdMaxInc-2}
\end{align}
\end{prop}
\begin{proof}
Assume that the maps with said properties exist, then the points $\iota(s_{ij}) \in K$ form a witness square which implies that the measurements are maximally incompatible.

If the measurements are maximally incompatible, then let $x_{ij} \in K$, $i, j \in \{0, 1\}$, be the corresponding witness square. Now define the maps $\iota: S \to K$ and $\Pi: K \to S$ as follows: $\iota$ is determined by the property that the images of the vertices of $S$ form the witness square:
\begin{equation}
\label{eq:ext_iot}
\iota(s_{ij}) = x_{ij}
\end{equation}
and for $x \in K$ we put
\begin{equation}
\label{eq:ext_pi}
\Pi(x) = f_\A (x) (s_{10} - s_{00}) + f_{\A'} (x) (s_{01} - s_{00}) + s_{00}.
\end{equation}
We have
\begin{equation*}
\Pi(\iota(s_{ij})) = \Pi(x_{ij}) = s_{ij}
\end{equation*}
and so $\Pi \circ \iota = \id$ follows as well. Note that Eq. \eqref{eq:GPT-comdMaxInc-1} and \eqref{eq:GPT-comdMaxInc-2} are satisfied simply because $\iota(s_{ij}) = x_{ij}$.
\end{proof}
The maps in Prop. \ref{prop:GPT-condMaxInc} can be uniquely extended to positive maps $\iota: A(S)^* \to A(K)^*$, $\Pi:
A(K)^* \to A(S)^*$, that is, linear maps preserving the positive cones. Note also that the map $\Pi$ is uniquely determined by the given pair of measurements while $\iota$ is given by the choice of the witness square, which may be non-unique.

\begin{coro}
\label{coro:maxinc}
Let $\A$, $\A'$ and $\B$, $\B'$ be two pairs of maximally incompatible measurements on a state space $K$ defined by the effects $f_\A, f_{\A'}, f_\B, f_{\B'}$. Let $\Pi_\A$, $\iota_\A$ and $\Pi_\B$, $\iota_\B$ be the corresponding maps. Then there is an affine map $T:K\to K$, such that $f_\B=f_\A\circ T$, $f_{\B'}=f_{\A'}\circ T$ and
\begin{equation*}
\iota_\A=T\circ \iota_\B,\ \Pi_\B=\Pi_\A\circ T.
\end{equation*}
\end{coro}

\begin{proof}
Put $T=\iota_\A\circ\Pi_\B$, all the properties are checked straightforwardly.
\end{proof}

\begin{remark}
\label{rem:maxinc_square}
Note that the above results imply that $\{\A_0, \A_1\}$ is, up to affine isomorphisms, the unique maximally incompatible pair of two-outcome measurements on $S$.
\end{remark}

\subsection{Tensor products and bipartite systems in GPTs}

To describe the state spaces of composite systems in GPTs, we need the notion of a tensor product of the state spaces. For simplicity, we will only consider the tensor product of a state space $K$ with itself.

There are several ways to define the tensor product of compact convex sets, but there is a minimal and a maximal one. All of the possible tensor products are compact convex subsets in the tensor product $A(K)^* \otimes A(K)^*$.

The minimal tensor product, denoted by $K \tmin K$, is the convex hull of the points of the form $x \otimes y$ for $x,y \in K$, i.e.
\begin{equation*}
K \tmin K = \conv( \{ x \otimes y: x, y \in K \}).
\end{equation*}
In other words, this is the smallest composite state space containing all locally prepared states. The maximal tensor product, denoted $K \tmax K$, is the state space of all non-signaling states, that is
\begin{align*}
K \tmax K = \{& \psi \in A(K)^* \otimes A(K)^*: \psi(f \otimes g) \geq 0 \\ &\forall f, g \in A(K)^+, (1\otimes 1) (\psi) = 1 \}.
\end{align*}
To give more insight into these definitions, we will look at state spaces of quantum theory (Example \ref{exm:quantum}). In this case, $\states_\Ha \tmin \states_\Ha$ is the set of all separable states and $\states_\Ha \tmax \states_\Ha$ is the set of all (normalized) entanglement witnesses \cite[Definition 6.38]{HeinosaariZiman-MLQT}.

Every GPT in which we want to describe bipartite systems must come equipped with a composition rule on how to form the joint state space of two (or more) systems. We will denote the joint state space as $K \treal K$ and it will represent the set of all bipartite states for the given system. Note that $K \treal K$ does not have any general definition as it is specified by the theory we are working with. Some properties of the composition rule are sometimes imposed (such as it is given by a symmetric monoidal structure), but here we only require that it is a state space such that
\begin{equation*}
K \tmin K \subseteq K \treal K \subseteq K \tmax K.
\end{equation*}

In quantum theory, the joint state space is $\states_\Ha \treal \states_\Ha = \states_{\Ha \otimes \Ha}$ so that all of the above inclusions are strict. On the other hand, for the classical bit (or any simplex) we have $\states_C \tmin \states_C=\states_C\tmax \states_C$, so that the joint state space is unique, denoted by $\states_C\otimes \states_C$. It can be easily seen that $\states_C\otimes \states_C$ can be identified with the 3-dimensional simplex, with vertices labeled by $\{(i,j),\ i,j=0,1\}$. In the Boxworld GPT, the composition rule is the maximal tensor product, so that $S\treal S=S\tmax S$. Here the minimal tensor product $S\tmin S$ is quite different from $S\tmax S$ (the GPT with the composition rule $S\treal S=S\tmin S$ is called GLT in \cite{Barrett-GPTinformation}).

Consider a state space $K$ and let $K\treal K$ be the joint state space. Let $\A$, $\B$ be two-outcome measurements on $K$ given by effects $f_A,f_B\in E(K)$ and let $\A \otimes \B$ denote the four-outcome measurement obtained by measuring $\A$ on the first part and $\B$ on the other. This measurement is determined by the effects $\{f_\A\otimes f_\B,(1-f_\A)\otimes f_\B, f_\A\otimes (1-f_\B), (1-f_\A)\otimes(1-f_\B)\}$ and it is clear that $K\tmax K$ is the largest state space such that all such locally prepared measurements are valid. Note also that for any $f \in E(K)$ and $x \in K \tmax K$ one can define $(f \otimes \id)(x) \in A(K)^{*+}$ as the unique functional such that for any other $g \in E(K)$ we have
\begin{equation*}
((f \otimes \id)(\phi))(g) = (f \otimes g)(\phi).
\end{equation*}
In particular, $(1\otimes id)(x)$ and $(id\otimes 1)(x)$ belong to $K$ and define the two marginals of $x$ which correspond to partial traces in quantum theory.

\subsection{Review of the CHSH inequality}
We provide a very short introduction to the CHSH inequality, in the setting of non-local boxes. These are defined as a black box, with two inputs with possible values $\A,\A'$ and $\B,\B'$, respectively, and two outputs, each with values 1 or -1. It is assumed that such a box describes the situation when two experimenters (Alice and Bob), each on their part of a bipartite system, apply one of a given pair of two-outcome measurements $\{\A,\A'\}$ resp. $\{\B,\B'\}$.

Any non-local box $x$ is fully described by the outcome probabilities $P_x(\epsilon, \eta | \C, \D)$, with $\C=\A$ or $\A'$, $\D=\B$ or $\B'$ and $\epsilon,\eta\in \{-1,1\}$. Assume that the measurements $\A$ and $\B$ are chosen. Then Alice will see the outcome $1$ on her part with the probability $P_x(1,1 | \A, \B) + P_x(1, -1 | \A, \B)$. For this to be a well-defined marginal outcome probability of $\A$, we require that it stays the same if the other measurement is $\B'$, that is,
\begin{equation}
\begin{aligned}
P_x (1, 1 | \A, \B) + P_x (1, -1 | \A, \B) = \\ = P_x (1, 1 | \A, \B') + P_x (1, -1 | \A, \B').
\end{aligned}
\label{eq:gpt-nonsignalingA}
\end{equation}
This condition, together with the analogical condition
\begin{align}
\begin{aligned}
P_x (1, 1 | \A, \B) + P_x (-1, 1 | \A, \B) = \\ = P_x (1, 1 | \A', \B) + P_x (-1, 1 | \A', \B)
\end{aligned}
\label{eq:gpt-nonsignalingB}
\end{align}
are called the non-signaling conditions \cite{Cirelson-bound, PopescuRohrlich-PRbox}, because they mean that neither side can signal to the other by only using different local measurements and without announcing the outcome of the measurement. Non-local boxes satisfying these conditions, also called non-signaling boxes, are of particular interest in the theory of Bell inequalities, see also \cite{BrunnerCavalcantiPironioScaraniWehner-review}.

It is clear that if $K$ is a state space in a GPT, then any $x\in K\treal K$ and any measurements $\{\A,\A'\}$ of the first part and $\{\B,\B'\}$ on the other implement a non-signaling box.

The central quantity for the formulation of the CHSH inequality is the correlation $E(\A,\B)$ between the measurements $\A$ and $\B$, defined as
\begin{align*}
E(\A,\B) &= P_x(1, 1 |\A, \B) - P_x(1, -1 | \A, \B) \\ &- P_x(-1, 1 | \A, \B) + P_x(-1, -1 | \A, \B).
\end{align*}
It is straightforward that we have $-1 \leq E(\A, \B) \leq 1$. For the two pairs of measurements $\A$, $\A'$ and $\B$, $\B'$, the CHSH quantity $X_\chsh$ is given as
\begin{equation*}
X_\chsh = E(\A, \B) + E(\A, \B') + E(\A', \B) - E(\A', \B') .
\end{equation*}
It is known that in classical theories we have $|X_\chsh| \leq 2$, this is called the the CHSH inequality. This inequality is violated in quantum theory, where the Tsirelson bound \cite{Cirelson-bound} gives $|X_\chsh| \leq 2 \sqrt{2}$. The maximal value reachable by a non-signaling theory coincides with the algebraic maximum, which is $|X_\chsh| = 4$.

The non-signaling box attaining the value $X_\chsh=4$ was defined by Popescu and Rohrlich \cite{PopescuRohrlich-PRbox}, it is called the PR-box. The PR-box is determined by the outcome probabilities
\begin{equation}
\label{eq:square_prob}
P_x(\epsilon,\eta| \C, \D) =
\begin{dcases}
\dfrac{1}{2} &
\begin{aligned}
&\text{if } \C\D \ne \A'\B' \text{ and } \epsilon \eta = 1 \\
&\text{or } \C\D=\A'\B' \text{ and } \epsilon \eta = -1
\end{aligned}
\\
0 & \text{otherwise}.
\end{dcases}
\end{equation}
It can be seen that the value $X_\chsh=-4$ is obtained from the PR-box by relabelling the outcomes on one of the sides
and that this is the only other possibility for maximal CHSH violation. We provide some more details in Appendix \ref{sec:PRbox} for the convenience of the reader.

\section{Implementations of the PR-box} \label{sec:finding}
The aim of this section is to characterize all implementations of the PR-box in the GPT framework, that is, for a state space $K$, we want to describe the states $\phi\in K\treal K$ and pairs of two-outcome measurements $\{\A,\A'\}$ and $\{\B,\B'\}$ such that the  corresponding outcome probabilities maximally violate the CHSH inequality.

The best known implementation of the PR-box is provided in the Boxworld GPT. It was shown in \cite{Barrett-GPTinformation} that all non-signaling boxes can be implemented with elements of the state space $S\tmax S$ and the outcome probabilities are obtained by applying the measurements given by $\A=\B=\A_0$ and $\A'=\B'=\A_{1}$. The state in $S\tmax S$ corresponding to the PR-box is given by
\begin{equation}
\label{eq:square_PR}
\phi_{S}=\frac12((s_{00}-s_{10})\otimes s_{00}+s_{11}\otimes s_{10}+s_{10}\otimes s_{01}),
\end{equation}
up to local isomorphisms, this is the only implementation of the PR-box on $S$, see Appendix \ref{sec:square} for a proof.

It is known that for some systems the degree of compatibility of measurements is tied to violation of the CHSH
inequality \cite{WolfPerezgarciaFernandez-measIncomp, StevensBusch-CHSHIncomp} and, as we have seen in Prop.
\ref{prop:GPT-condMaxInc}, all maximally incompatible pairs are obtained from the square state space by embedding it
into other state spaces. It is therefore not surprising that all implementations of the PR-box are obtained from the above implementation on $S$.

\begin{theorem}
\label{thm:finding-result}
Let $K$ be a state space and let $\{\A,\A'\}$, $\{\B,\B'\}$ and $\phi \in K \treal K$ be an implementation of the PR-box. Then both pairs $\{\A, \A'\}$ and $\{\B,\B'\}$ are maximally incompatible. Moreover, let $\iota_\A$, $\Pi_\A$ and $\iota_\B$, $\Pi_\B$ be the maps for these pairs given by Prop. \ref{prop:GPT-condMaxInc}. Then
\begin{equation*}
\phi = (\iota_\A \otimes \iota_\B)(\phi_S) + \phi^\perp
\end{equation*}
where $\phi_S \in S \tmax S$ is given by \eqref{eq:square_PR} and $\phi^\perp\in \ker(\Pi_\A)\otimes \ker(\Pi_\B)$.
\end{theorem}

\begin{proof}
The proof is in Appendix \ref{sec:proof}.
\end{proof}

We see from this result that maximally incompatible pairs of measurements are necessary for maximal CHSH violation, moreover, having such pairs of measurements on both sides, we can construct all possible states. While existence of such measurements is a property of $K$, we need also a joint state space $K\treal K$ that contains at least one of the candidate states. It follows from the next result that if we work with the maximal tensor product, maximal incompatibility is also sufficient for existence of an implementation of the PR-box.

\begin{prop}
\label{prop:square-max}
Let $K$ be a state space on which there exists a pair of maximally incompatible measurements and let $\iota:S\to K$ be the corresponding map. Then
\begin{equation*}
(\iota\otimes \iota)(\phi_S)\in K \tmax K.
\end{equation*}
\end{prop}
\begin{proof}
Since $\iota$ extends to a positive map $A(S)^*\to A(K)^*$ with the respective positive cones, the assertion follows by the fact that positive maps are completely positive for the maximal tensor product. In more details: note that the adjoint map $\iota^*:A(K)\to A(S)$, given by \[ \iota^*(f)(s)= f(\iota(s)), \quad s\in S,\ f\in A(K) \] is again positive. Hence, for any $h_1,h_2 \in A(K)^+$ we have
\begin{equation*}
(h_1 \otimes h_2)( (\iota \otimes \iota)( \phi_S ) ) = ((\iota^* h_1) \otimes (\iota^* h_2))( \phi_S ) \geq 0,
\end{equation*}
as $\phi_S \in S \tmax S$ and $\iota^* h_i \in A(S)^+$. It follows that $(\iota \otimes \iota)( \phi_S ) \in K \tmax K$.
\end{proof}

We finish this section by observing that the element $(\iota\otimes \iota)(\phi_S)$, corresponding to the special case $\phi^\perp=0$ in Thm. \ref{thm:finding-result}, can be constructed from a witness square $\{x_{ij}\}$ of the maximally incompatible pair of measurements as \[ (\iota\otimes \iota)(\phi_S)= \dfrac{1}{2} ( ( x_{00} - x_{10} ) \otimes x_{00} + x_{11} \otimes x_{10} + x_{10} \otimes x_{01} ). \] This is immediate from \eqref{eq:ext_iot} and \eqref{eq:square_PR}. It is clear that we can construct a similar state for different measurements on each side, using some respective witness squares.

\section{Classical channels as implementations of the PR-box} \label{sec:cPR}

In this section we discuss known results, showing that non-signaling classical channels implement the PR-box.

As we have seen in Example \ref{exm:clchannels}, the set of classical bit channels is isomorphic to the square state space. By a similar reasoning as in the proof of Proposition \ref{prop:square-max}, we can extend this to an isomorphism  of the maximal tensor products $\Ce_C\tmax \Ce_C$ and $S\tmax S$. In this way, we may view the Boxworld as a GPT based on classical channels. It this section, we will describe the PR-box implementation \eqref{eq:square_PR} in this setting.

It is immediate from the definition that the non-local boxes can be identified with classical bipartite channels $\states_C\otimes \states_C\to \states_C\otimes\states_C$, where the four vertices of $\states_C\otimes\states_C$ are labeled by elements of $\{\A,\A'\}\times\{\B,\B'\}$ in the input space and by elements of $\{1,-1\}^2$ in the output. Since $S\tmax S$ is identified with the non-signaling boxes, we see from the above remarks that $\Ce_C\tmax \Ce_C$ can be described as the set of non-signaling classical bipartite channels.

As shown in Example \ref{exm:clchannels}, the measurements on $\Ce_C$ corresponding to $\A_0$, $\A_1$ on $S$ are given by the effects $F_{s_0,\pi}$ and $F_{s_1,\pi}$, respectively (where $\pi$ is as in Example \ref{exm:clbit}).  Note also that the corresponding maps $\iota$ and $\Pi$ are precisely the isomorphism $S\to \Ce_C$ and its inverse. To construct the bipartite channel $\Phi_C$ corresponding to the state $\phi_S$ of \eqref{eq:square_PR}, we first find the elements of $\Ce_C$ corresponding to the vertices of $S$. Looking at Example \ref{exm:clchannels}, the channels $\Phi_{ij}\simeq s_{ij}$ are determined by
\begin{align*}
&\pi(\Phi_{ij}(s_0))=i,
&&\pi(\Phi_{ij}(s_1))=j,
\end{align*}
hence $\Phi_{00}=1_{\states_C}(\cdot)s_0$ and $\Phi_{11}=1_{\states_C}(\cdot)s_1$ are constant channels, $\Phi_{01}=id$ and $\Phi_{10}$ is the negation channel, given by $s_0\mapsto s_1$ and $s_1\mapsto s_0$. It is now easily checked that $\Phi_C$ is given by
\begin{align}
\Phi_C(s_0\otimes s_0)&=\Phi_C(s_0\otimes s_1)=\Phi_C(s_1\otimes s_0)\notag\\
&=\frac12(s_0\otimes s_0+s_1\otimes s_1)\label{eq:clchannel_PRsym}\\
\Phi_C(s_1\otimes s_1)&=\frac12(s_0\otimes s_1+s_1\otimes s_0)\label{eq:clchannel_PRanti}
\end{align}
It can be also checked directly that this channel, together with the above pair of measurements applied on both sides, maximally violates the CHSH inequality. The protocol is depicted on Fig. \ref{fig:cPR-protocol}. As in the case of $S$, this is the only implementation of the PR-box in this setting, up to local isomorphisms.

\begin{figure}
\includegraphics[width=.8\linewidth]{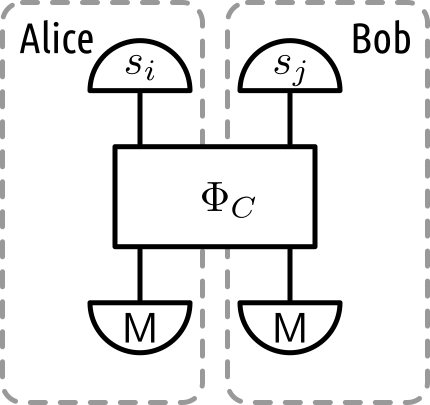}
\caption{Alice and Bob can use the channel $\Phi_C$ to maximally violate the CHSH inequality as follows: Alice inputs the state $s_i$ and Bob inputs the state $s_j$ where $i, j \in \{ 0, 1 \}$. This choice exactly corresponds to the choice of channel measurement $\A=\B$ or $\A'=\B'$ on either side. Finally both Alice and Bob apply the measurement $\M$ corresponding to the effect $\pi \in E(\states_C)$. Note that for both Alice and Bob the procedure of inputting a testing state into the channel and measuring the outcome comprises a single measurement on the bipartite channel itself. Also note that both Alice and Bob can only access their respective input and output and since the channel $\Phi_C$ is non-signaling they can not obtain any information about each others choices of channel measurement. \label{fig:cPR-protocol}}
\end{figure}

\section{Quantum channels as implementations of the PR-box} \label{sec:qPR}

We now get to the most important example of a state space for this work. Let $\Ha$ be a finite dimensional Hilbert space and let $\Ce(\Ha)$ denote the set of all quantum channels on $B(\Ha)$, that is, all completely positive and trace preserving linear maps $B_h(\Ha) \to B_h(\Ha)$. Let $\{ |i\> \}_{i=1}^{\dH}$ be an orthonormal basis of $\Ha$, then $|\psi^+_{\dH}\> = \sum_{i=1}^{\dH} |i\> \otimes |i\> \in \Ha \otimes \Ha$ is a multiple of the maximally entangled state. The Choi matrix of a channel $\Phi$ is given as $C(\Phi) = (\Phi \otimes id)(|\psi^+_{\dH}\>\<\psi^+_{\dH}|)\in B(\Ha\otimes \Ha)$. The set $\Ce(\Ha)$ is isomorphic to the set of Choi matrices
\begin{equation*}
\Choi(\Ha): = \{ A \in B_h(\Ha \otimes \Ha): A \geq 0, \Tr_1(A) = \I \}
\end{equation*}
where $\Tr_1$ is the partial trace over the first Hilbert space. Clearly, $\Ce(\Ha)$ is a compact convex subset of the finite dimensional real vector space of linear maps on $B_h(\Ha)$ and can therefore be treated as a state space in some GPT.  The following description of the quantum channel GPT is based on the ideas of \cite{ChiribellaDArianoPerinotti-PPOVM}.

The measurements of quantum channels are in principle similar to measurements of classical channels described in Example \ref{exm:clchannels}. Clearly we can input any quantum state into the channel in question and measure the outcome, but this does not describe all of the possible measurements on quantum channels for a simple reason: the input state can also be entangled to other system. It turns out that all measurements on quantum channels are described as procedures where we input a bipartite and potentially entangled state $\rho \in B_h(\Ha \otimes \Ha')$ into the channel (possibly tensored with identity) and measure the outcome. The measurement on channels given by $\rho$ and a quantum measurement described by the POVM $E_1, \ldots, E_n$ is determined by the outcome probabilities
\begin{equation*}
P_\Phi (i | \E) = \Tr( (\Phi \otimes \id)(\rho) E_i )
\end{equation*}
for $i \in \{1, \ldots, n \}$ and any quantum channel $\Phi: B_h(\Ha) \to B_h(\Ha)$. Any effect on channels has the form
\begin{equation*}
F_{\rho, E}(\Phi) :=\Tr( (\Phi\otimes \id)(\rho)E )
\end{equation*}
for some effect $E$ and input state $\rho$. Note also that this expression is not unique. In particular, we can always assume the state $\rho$ to be pure as we can always purify it by enlarging the Hilbert space $\Ha'$.

We can describe the measurements in an equivalent way, using the Choi matrices. Then there are operators $F_1, \ldots, F_n \in B_h(\Ha \otimes \Ha)$ such that
\begin{equation*}
P_\Phi (i | \E) = \Tr( C(\Phi) F_i )
\end{equation*}
for $i \in \{1, \ldots, n \}$. One can show that we may always choose $F_i \geq 0$ for all $i \in \{1, \ldots, n \}$ and we must have $\sum_{i=1}^n F_i = \I \otimes \sigma$ where $\sigma \in \states_\Ha$. Such a collection of operators $F_1, \ldots, F_n$ is called a process POVM, or PPOVM, or a quantum tester, which were first introduced in \cite{Ziman-ppovm, ChiribellaDArianoPerinotti-PPOVM}. This passage from quantum channels to Choi matrices is directly related to viewing the quantum channels in the GPT picture, see also Fig. \ref{fig:qc}.

\begin{figure}[ht]
\centering
\includegraphics[angle=270]{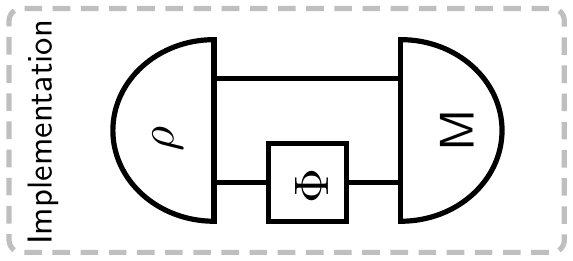}
\includegraphics[angle=270]{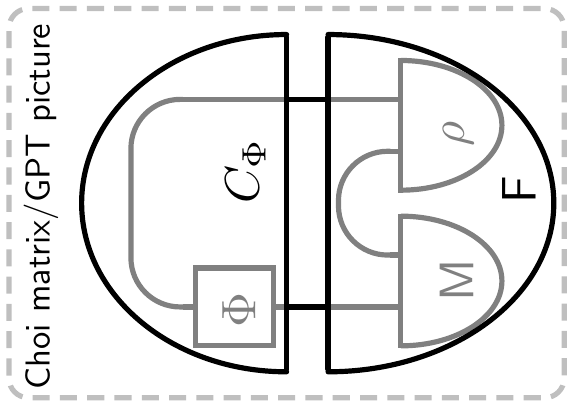}
\caption{
Desciption of a measurement on quantum channels. On the left, we have an implementation of the measurement with an input state $\rho$ of the system composed with an ancilla and the final measurement $\M$, which is a usual quantum measurement described by a POVM. Note that it is not necessary to consider convex combinations, since these can be included into the ancilla. The right side shows the same measurement represented by the PPOVM $\F$ applied on the Choi matrix $C_\Phi$, obtained from $\Phi$ by applying the channel to (one part) of the state $|\psi^+_{\dH}\>$. This is depicted as the ''bending of the input wire'' of $\Phi$. Similarly, the PPOVM $\F$ is the \emph{link product} of $\rho$ and $\M$ \cite{ChiribellaDArianoPerinotti-PPOVM} which is obtained by application of $|\psi^+_{\dH}\>$ on parts of the inputs, as shown in the picture. The representation of $\Phi$ by its Choi matrix $C_\Phi$ can be straightforwardly interpreted as  viewing the set of channels in the GPT picture.  \label{fig:qc}}
\end{figure}

We now need to specify the joint state space $\Ce(\Ha) \treal \Ce(\Ha)$. This will be defined as the set of completely positive maps in $\Ce(\Ha) \tmax \Ce(\Ha)$, which is precisely the set $\Ce^{ns}(\Ha\otimes\Ha)$ of bipartite non-signaling (or causal) channels, as defined in \cite{BeckmanGottesmanNielsenPreskill-channels}. On the set of Choi matrices, it can be seen that the set of positive elements in $\Choi(\Ha)\tmax \Choi(\Ha)$ is isomorphic to the set of Choi matrices of elements in $\Ce^{ns}(\Ha\otimes\Ha)$, via the map
\begin{equation*}
\id_1\otimes \swap_{23}\otimes \id_4: B(\Ha^{\otimes 4})\to B(\Ha^{\otimes 4}),
\end{equation*}
where $\swap$ denotes the swap gate, that is for $X, Y \in B(\Ha)$ we have
\begin{equation*}
\swap(X \otimes Y) = Y \otimes X.
\end{equation*}

We next describe a pair of maximally incompatible two-outcome measurements on $\Ce(\Ha)$. Let $\rho_1, \rho_2, \sigma_1, \sigma_2 \in \states_\Ha$, such that $\rho_1 \rho_2 = \sigma_1 \sigma_2 = 0$. Then there exist projections $M, N$, such that $M \rho_1 = \rho_1$, $M \rho_2 = 0$, $N \sigma_1 = \sigma_1$, $N \sigma_2 = 0$. Also let us denote $M^\perp = \I - M$, $N^\perp = \I - N$. Let the two-outcome measurements $\C_1$, $\C_2$ be given by the effects $F_{\sigma_1, M}$ and $F_{\sigma_2,M}$. Consider the channels $B_h(\Ha) \to B_h(\Ha)$ given for $X \in B_h(\Ha)$ as
\begin{align*}
\Phi_{00} (X) &= \Tr(X) \rho_2, \\
\Phi_{10} (X) &= \Tr(N X) \rho_1 + \Tr( N^\perp X) \rho_2, \\
\Phi_{01} (X) &= \Tr(N X) \rho_2 + \Tr( N^\perp X) \rho_1, \\
\Phi_{11} (X) &= \Tr(X) \rho_1
\end{align*}
It is straightforward to verify that the measurements $\C_1$, $\C_2$ are maximally incompatible and the channels $\Phi_{00}$, $\Phi_{10}$, $\Phi_{01}$, $\Phi_{11}$ form a witness square for $\C_1$, $\C_2$. Let $\iota: S\to \Ce(\Ha)$ be the map \eqref{eq:ext_iot} for this witness square.

Let $\Phi=(\iota\otimes \iota)(\Phi_S)$ be the tensor product element as described in Prop. \ref{prop:square-max}. Let us denote $\rho_{cor} = \frac{1}{2}( \rho_1 \otimes \rho_1 + \rho_2 \otimes \rho_2 )$ and $\rho_{ac} = \frac{1}{2}( \rho_1 \otimes \rho_2 + \rho_2 \otimes \rho_1 )$. One can check that for $X, Y \in B_h(\Ha)$, we have
\begin{align*}
\Phi (X \otimes Y) &= \dfrac{1}{2} \big( ( \Phi_{00}(X) - \Phi_{10}(X) ) \otimes \Phi_{00} (Y) \\
&+ \Phi_{11}(X) \otimes \Phi_{10}(Y) + \Phi_{10}(X) \otimes \Phi_{01}(Y) \big) \\
&= \Tr( (N^\perp \otimes N^\perp) (X \otimes Y) ) \rho_{cor} \\
&+ \Tr( (\I \otimes \I - N^\perp \otimes N^\perp) (X \otimes Y) ) \rho_{ac}.
\end{align*}
By linearity, for every $\rho \in S_{\Ha \otimes \Ha}$ we have
\begin{align}
\Phi(\rho) &= \Tr( (N^\perp \otimes N^\perp) (\rho) ) \rho_{ac} \nonumber \\
&+ \Tr( (\I \otimes \I - N^\perp \otimes N^\perp) (\rho) ) \rho_{cor}. \label{eq:qPR-channel}
\end{align}
It is easy to see that this is indeed a quantum channel, so that we have $\Phi\in \Ce^{ns}(\Ha\otimes \Ha)=C(\Ha)\treal C(\Ha)$. It follows by Thm. \ref{thm:finding-result}, but is also straightforward to verify, that the channel $\Phi$ and the measurements $\A=\B=\C_1$ and $\A'=\B'=\C_2$ are an implementation of the PR-box. The protocol is similar to the one used in the case of classical channels: Alice and Bob share the bipartite channel $\Phi$ and they both can choose the input state $\sigma_1$ or $\sigma_2$ each on their part of the channel, followed by the measurement $\{M, M^\perp\}$ on their part of the output.

\begin{exm}
\label{exm:PhiS} The above pair of maximally incompatible measurements is a generalization of an example studied in
\cite{SedlakReitznerChiribellaZiman-compatibility, JencovaPlavala-maxInc}, with $\dH=2$ and
$N=M=\rho_1=\sigma_1=|0\>\<0|$ and $N^\perp=M^\perp=\rho_2=\sigma_2=|1\>\<1|$. The corresponding channel
has the form
\begin{align*}
\Phi(\rho) &= \<11|\rho|11\>\frac12(|01\>\<01|+|10\>\<10|) \\
&+ (1-\<11|\rho|11\>)\frac12(|00\>\<00|+|11\>\<11|).
\end{align*}
The resulting implementation of the PR-box was already observed in \cite{BeckmanGottesmanNielsenPreskill-channels, HobanSainz-channels, PlavalaZiman-PRbox}. Note that under the identification $S\simeq \Ce_C$, the map $\iota$ becomes the inclusion of $\Ce_C$ onto the set of classical-to-classical qubit channels determined by $|i\>\<i|\mapsto \sum_j P(j|i)|j\>\<j|$ for conditional propabilities $P(j|i)$, while $\Pi$ is a projection of $\Ce(\Ha)$ onto this set.

\end{exm}

As the results of Thm. \ref{thm:finding-result} suggest, even if the pair of maximally incompatible measurements is fixed, there can be more bipartite non-signaling channels that implement the PR-box. Any such channel will be called a PR-channel.

Observe that the PR-channel $\Phi$ given by \eqref{eq:qPR-channel} is of a special form, called an entanglement-breaking channel. An entanglement-breaking channel is such that for any state $\omega \in \states_{\Ha \otimes \Ha \otimes \Ha \otimes \Ha}$ we have that $(\Phi \otimes \id)(\omega)$ is a separable state. Such channels are also called measure-and-prepare, since we first perform a measurement, in this case the two-outcome measurement given by the effect $N^\perp\otimes N^\perp$, and according to the result, we prepare one of a given set of states, in this case $\rho_{ac}$ or $\rho_{cor}$.

The structure of the channel $\Phi$ is even more simple. In fact, since $N$ is a projection and the two states $\rho_{ac}$ and  $\rho_{cor}$ commute, $\Phi$ is classical-to-classical. Moreover, note that both the measurement and the target states are separable. Even so, the resulting channel implements a PR-box.

Now we would like to see whether there are PR-channels of a more complicated structure. It would be quite hard to characterize all such channels in the general case. For this reason, in the next section we restrict to $\dH = 2$, i.e. to qubits.

\section{The special case of qubits} \label{sec:qubit}
In this section we restrict to qubits, i.e. $\dH = 2$ and we characterize all PR-channels that exist in this scenario. So let $\Ha$ be a complex Hilbert space, $\dH = 2$ and let $\{ |0\>, |1\> \}$ be an orthonormal basis of $\Ha$. For vectors from $\Ha \otimes \Ha$, we will use the shorthand $|i\> \otimes |j\> = |ij\>$ for $i, j \in \{0, 1\}$. We will also use the notation $\oplus$ for addition modulo 2.

We will begin by characterizing all maximally incompatible pairs of two-outcome measurements and their witness squares. Let the two measurements be given by effects $F_{\rho,M}$ and $F_{\sigma,N}$, for pure states $\rho, \sigma \in \states_{\Ha \otimes \Ha'}$ and $M,N\in E(\Ha\otimes \Ha')$. Let $\ker(M)$ denote the projection onto the kernel of $M$. Then by \eqref{eq:witnessquare2} and \eqref{eq:witnessquare3}, a witness square $\Phi_{00}, \Phi_{10}, \Phi_{01}, \Phi_{11} \in \Ce(\Ha)$ must satisfy
\begin{align*}
\ker(M) (\Phi_{10} \otimes \id)(\rho) \ker(M) &= 0, \\
\ker(M) (\Phi_{11} \otimes \id)(\rho) \ker(M) &= 0,
\end{align*}
and
\begin{align*}
\ker(M) (\Phi_{00} \otimes \id)(\rho) \ker(M) &= (\Phi_{00} \otimes \id)(\rho), \\
\ker(M) (\Phi_{01} \otimes \id)(\rho) \ker(M) &= (\Phi_{01} \otimes \id)(\rho).
\end{align*}
From \eqref{eq:witnessquare1} it follows that
\begin{equation*}
(\Phi_{00} \otimes \id)(\rho) + (\Phi_{11} \otimes \id)(\rho) = (\Phi_{10} \otimes \id)(\rho) + (\Phi_{01} \otimes \id)(\rho),
\end{equation*}
and after applying $\ker(M)$ we get
\begin{equation}
(\Phi_{00} \otimes \id)(\rho) = (\Phi_{01} \otimes \id)(\rho). \label{eq:qubit-rhoKerM}
\end{equation}
It follows that $\rho$ cannot have maximal Schmidt rank as then \eqref{eq:qubit-rhoKerM} would imply $\Phi_{00} = \Phi_{01}$ which is impossible by \eqref{eq:witnessquare2} and \eqref{eq:witnessquare3}. Since we have assumed $\dH = 2$, it follows that $\rho$ must have Schmidt rank $1$, i.e. $\rho$ must be a pure product state, so we can assume $\rho = |x \> \< x |$ for some $|x \> \in \Ha$, $\Vert x \Vert = 1$, and $M \in B_h(\Ha)$, since $\rho$ is not entangled and therefore we do not need the ancillary Hilbert space $\Ha'$. From
\begin{equation*}
\Tr( \Phi_{10}(|x \> \< x |) M ) = \Tr( \Phi_{11}(|x \> \< x |) M ) = 1
\end{equation*}
and
\begin{equation*}
\Tr( \Phi_{00}(|x \> \< x |) M ) = \Tr( \Phi_{01}(|x \> \< x |) M ) = 0
\end{equation*}
it follows that there must be an orthonormal basis $|\xi_0\>, |\xi_1\>$ of $\Ha$ such that
\begin{align*}
&M = | \xi_0 \> \< \xi_0 |,
&&M^\perp = | \xi_1 \> \< \xi_1 |,
\end{align*}
and for $i, j \in \{0, 1 \}$ we have
\begin{equation*}
\Phi_{ij}(|x \> \< x |) = | \xi_{i\oplus 1} \> \< \xi_{i\oplus 1} |.
\end{equation*}
In a similar fashion, one can show that we must have $\sigma = |y \> \< y |$ for some $| y \> \in \Ha$, $\Vert y \Vert = 1$, and that there is an orthonormal basis $| \eta_0 \>, |\eta_1 \>$ of $\Ha$ such that
\begin{align*}
&N = | \eta_0 \> \< \eta_0 |,
&&N^\perp = | \eta_1 \> \< \eta_1 |,
\end{align*}
and for $i, j \in \{0, 1 \}$ we have
\begin{equation*}
\Phi_{ij}(|y \> \< y |) = | \eta_{j\oplus 1} \> \< \eta_{j\oplus 1} |.
\end{equation*}

Let $L\in B(\Ha)$ be given by $L|0\>=|x\>$, $L|1\>=|y\>$ and let $\Phi^L_{ij}=\Phi_{ij}(L\cdot L^*)$. Then $\Phi^L_{ij}$ are completely positive maps satisfying \eqref{eq:witnessquare1}, with Choi matrices
\begin{equation*}
C(\Phi^L_{ij}) =
\begin{pmatrix}
|\xi_{i\oplus 1}\>\<\xi_{i\oplus 1}| & X_{ij} \\
X_{ij}^* & |\eta_{j\oplus 1}\>\<\eta_{j\oplus 1}|
\end{pmatrix}
\end{equation*}
where $X_{ij} \in B(\Ha)$. It follows by positivity of $C(\Phi^L_{ij})$ that we must have
\begin{equation*}
X_{(i\oplus 1)(j\oplus 1)} = z_{ij} |\xi_i\>\<\eta_{j}|
\end{equation*}
where $z_{ij}\in \mathbb C$, $|z_{ij}|\le 1$, see \cite[1.3.2 Proposition]{Bhatia-PDM}. From Eq. \eqref{eq:witnessquare1} we get
\begin{equation*}
z_{00} |\xi_0 \>\<\eta_0| + z_{11} |\xi_1 \>\<\eta_1| = z_{10} |\xi_1 \>\<\eta_0| + z_{01} |\xi_0 \>\<\eta_1|
\end{equation*}
from which it follows that $z_{ij} = 0$ for all $i,j \in \{0, 1\}$. This implies that $C(\Phi^L_{ij})$ are block-diagonal matrices. In particular,
\begin{equation*}
0=\Tr (\Phi^L_{ij}(|0\>\<1|))=\Tr \Phi_{ij}(|x\>\<y|) = \<x|y\>,
\end{equation*}
so that $\{|x\>,|y\>\}$ is an orthonormal basis of $\Ha$ and we have proved that $C(\Phi_{ij})$ are block-diagonal in this basis. Applying unitary transformations to the input resp. output space, transforming the basis $|x\>,|y\>$ resp. $|\xi_0\>,|\xi_1\>$ to $|0\>,|1\>$, we can summarize as follows.

\begin{prop}
\label{prop:max_inc} Let $\A$, $\A'$ be two-outcome measurements on qubit channels, given by PPOVMs
$\{F_\A,F_\A^\perp\}$ and $\{F_{\A'},F_{\A'}^\perp\}$. Then $\A$, $\A'$ are
maximally incompatible if and only if, up to unitary conjugation on the input and output spaces,
\begin{align*}
&F_\A = |00\>\<00|,
&&F_\A^\perp=|10\>\<10| \\
&F_{\A'}=|\eta_0\>\<\eta_0|\otimes |1\>\<1|,
&&F_{\A'}^\perp=|\eta_1\>\<\eta_1|\otimes |1\>\<1|
\end{align*}
where $\{|\eta_0,|\eta_1\>\}$ is an orthonormal basis of $\Ha$. Moreover, there is a unique witness square for $\A,\A'$, with Choi matrices of the form
\begin{equation*}
C_{ij} =
\begin{pmatrix}
|i\oplus 1\>\<i\oplus 1| & 0\\
0 & |\eta_{j\oplus 1}\>\<\eta_{j\oplus 1}|
\end{pmatrix}
.
\end{equation*}
\end{prop}

By the above proposition, essentially any maximally incompatible pair is characterized by the choice of an orthonormal basis $|\eta_0\>,|\eta_1\>$. The corresponding channel measurements then consist of inputting $|0\>$ into the channel and measuring the basis $\{|i\>\}$ on the output, or inputting $|1\>$ and measuring the basis $\{|\eta_i\>\}$.

Let us pick some choice of the bases $|\eta^A_0\>,|\eta^A_1\>$ and $|\eta^B_0\>,|\eta^B_1\>$ on Alice's and Bobs part, respectively, and let $\A,\A'$ and $\B,\B'$ denote the corresponding maximally incompatible pairs. Note that the set of all PR-channels that give an implementation of the PR-box with these measurements is a face of $\Ce^{ns}(\Ha\otimes\Ha)$. As we will see, all such faces consist entirely of entanglement-breaking channels.

Let $\Phi$ be a qubit PR-channel and let $C$ be its Choi matrix. Then $C\in B(\Ha_{\out}\otimes \Ha_{\inn})$, where both the input and the output spaces are composed of Alice's and Bob's part: $\Ha_{\inn}=\Ha_{A,in}\otimes \Ha_{B,in}$ and $\Ha_{\out}=\Ha_{A,out}\otimes \Ha_{B,out}$. We write $C$ as a block matrix \[ C=\sum_{\alpha\in \{0,1\}^2} C_{\alpha,\beta}\otimes |\alpha\>\<\beta|_{\inn}, \] where $C_{\alpha,\beta}\in B(\Ha_{\out})$. To describe the structure of $C$, we need to introduce the following notations. For $x,y,z\in \mathbb C$, we denote
\begin{align*}
B_{\diag}(z) &=
\begin{pmatrix}
1 & z \\ \bar{z} & 1
\end{pmatrix}
= I_2+ z|0\>\<1|+\bar z|1\>\<0| \\
B_{\off}(x,y) &=
\begin{pmatrix}
0 & x \\ y & 0
\end{pmatrix}
= x |0\>\<1|+y|1\>\<0|,
&&i\ne j.
\end{align*}
For $r\in \mathbb N$, we denote by $B_r$ a block matrix in $B(\mathbb C^2\otimes \mathbb C^r)$ of the form
\begin{align}
B_r = \dfrac{1}{2} \Bigg( \sum_{p=1}^r &B_{\diag}(z_p)\otimes |p\>\<p| \nonumber \\ + \sum_{q\ne p=1}^r &B_{\off}(x_{p,q},y_{p,q})\otimes |p\>\<q| \Bigg).\label{eq:qubit_B}
\end{align}

\begin{prop}
\label{prop:PRchan}
Let $\Phi$ be a qubit PR-channel. Then there are isometries $U_\alpha:\mathbb C^2\to \Ha_{\out}$, $\alpha\in \{0,1\}^2$ and a decomposition $\{0,1\}^2=\Delta_0\cup \Delta_1$, such that $U_{\alpha}=U_{\beta}=:V_{\out}$ for $\alpha,\beta\in \Delta_1$ and $C$ has the form
\begin{align}
C = &(V_{\out}\otimes V_{\inn})B_r(V^*_{\out}\otimes V^*_{\inn}) \nonumber \\ &+ \sum_{\alpha\in \Delta_0}U_{\alpha}B_{\diag}(z_\alpha) U_{\alpha}^* \otimes |\alpha\>\<\alpha|\label{eq:PRform},
\end{align}
here $|z_\alpha|\le 1$, $r=|\Delta_1|$ and $V_{\inn}:\mathbb C^r\to \Ha_{\inn}$ is an isometry such that $V_{\inn}|p\>=|\alpha_p\>$, $\alpha_p\in \Delta_1$, $p=1,\dots,r$.
\end{prop}

The proof of this proposition is given in Appendix \ref{sec:PRproof}.

\begin{lemma}
\label{lemma:qubit_lemma3}
Let $r \leq 3$ and let $B_r$ be a matrix of the form \eqref{eq:qubit_B}. If $B_r$ is positive, then it is separable.
\end{lemma}

\begin{proof}
Let $B_r$ be positive. Since $B_r\in B(\mathbb C^2\otimes \mathbb C^r)$ and $r\le 3$, we may apply the PPT criterion \cite{Woronowicz-PPT, HorodeckiHorodeckiHorodecki-PPT}, that is, $B_r$ is separable if and only if it remains positive under partial transpose. We will apply the transpose to the first part, so we will show that the matrix
\begin{align*}
B^\Gamma_r = \dfrac{1}{2} \Bigg( \sum_{p=1}^r &B_{\diag}(z_p)^T\otimes |p\>\<p|\\ + \sum_{q\ne p=1}^r &B_{\off}(x_{p,q},y_{p,q})^T\otimes |p\>\<q| \Bigg)
\end{align*}
is positive. Let $V \in B(\mathbb C^2)$ be given as
\begin{equation*}
V =
\begin{pmatrix}
0 & 1 \\ 1 & 0
\end{pmatrix}
\end{equation*}
then $V$ is unitary, $V = V^*$ and for any $t \in \mathbb{R}$ and $z_1, z_2 \in \mathbb{C}$ we have
\begin{equation*}
\begin{pmatrix}
0 & 1 \\ 1 & 0
\end{pmatrix}
\begin{pmatrix}
t & z_1 \\ z_2 & t
\end{pmatrix}
\begin{pmatrix}
0 & 1 \\ 1 & 0
\end{pmatrix}
=
\begin{pmatrix}
t & z_2\\ z_1 & t
\end{pmatrix}
=
\begin{pmatrix}
t & z_1 \\ z_2 & t
\end{pmatrix}
^T
.
\end{equation*}
It follows that
\begin{equation*}
B_r^\Gamma = (V \otimes \I_r) B_r (V \otimes \I_r)
\end{equation*}
so we have $B_r^\Gamma \geq 0$ and $B_r$ is separable.
\end{proof}

We now prove the main result of this section.
\begin{theorem}
\label{thm:qubit-ETBchannel}
Let $\Phi$ be a qubit PR-channel. Then $\Phi$ is an entanglement-breaking channel.
\end{theorem}

\begin{proof}
The channel $\Phi$ is entanglement-breaking if and only if its Choi matrix $C=C(\Phi)$ is separable. The assertion now follows by Prop. \ref{prop:PRchan} and Lemma \ref{lemma:qubit_lemma3}.
\end{proof}

We will proceed by presenting a few examples of PR-channels. We concentrate on the choice $|\eta^A_i\>=|\eta^B_i\>=|i\>$. According to Appendix \ref{sec:PRproof}, in this case we have \[ U_\alpha|i\>=
\begin{dcases}
|ii\> & \text{if } \alpha\ne 11\\
|i(i\oplus 1)\> & \text{otherwise}
\end{dcases}
\] $r=3$ and $\Delta_0=\{11\}$. So any such channel is specified by the choice of the parameters $z_p, x_{p,q}, y_{p,q}$, $p,q=1,\dots,3$ such that the matrix $B_3$ is positive, and any choice of $z_4:=z_{11}$ with $|z_{4}|\le 1$.

An obvious choice is setting all these parameters to 0, in which case we obtain the channel of Example \ref{exm:PhiS}.

\begin{exm}
Another possible choice of parameters is
\begin{equation*}
z_p=\pm 1,\qquad x_{p,q}=y_{p,q}=0,\ \forall p,q.
\end{equation*}
The resulting channels $\Phi_\pm$ are similar to $\Phi_{S}$. We have
\begin{equation*}
\Phi_\pm (\rho) = (1 - \<11|\rho|11\> ) |\phi^\pm\>\<\phi^\pm| + \<11|\rho|11\> |\psi^\pm\>\<\psi^\pm|
\end{equation*}
where
\begin{align*}
|\phi^\pm\> &= \dfrac{1}{\sqrt{2}} ( |00\> \pm |11\> ), \\
|\psi^\pm\> &= \dfrac{1}{\sqrt{2}} ( |01\> \pm |01\> ).
\end{align*}
These channels are again classical-to-classical, but here the target states are pure and maximally entangled. A similar channel was also constructed by \cite{Crepeau-CHSH}.
\end{exm}

\begin{exm}
Let $W:\mathbb C^2\to V_{\out}$ be the isometry given by $W|i\>\mapsto |ii\>$ and let $\tilde W:=(V\otimes \I)W$. Let $M_\ell\in B(\Ha\otimes\Ha)$, $\ell=1,\dots,k$ be effects such that $\sum_\ell M_\ell=I-|11\>\<11|$ and let $|w_\ell|\le 1$, $\ell=0,1,\dots,k$. Then the entanglement-breaking channel
\begin{align*}
\Phi(\rho)= \frac12\biggl(\<11|\rho|11\>& \tilde W B_{\diag}(w_0)\tilde W^*\\
&+\sum_{\ell=1}^k \Tr (M_\ell\rho) WB_{\diag}(w_\ell)W^*\biggr)
\end{align*}
is a PR-channel of the required form, with values of the parameters $z_4=\omega_0$ and for $p,q\le 3$
\begin{align*}
z_p &= \sum_{\ell=1}^k w_\ell \<\alpha_p|M_\ell|\alpha_p\> \\ x_{p,q} = y_{p,q} &= \sum_{\ell=1}^k w_\ell \<\alpha_p|M_\ell|\alpha_q\>.
\end{align*}
This example contains the above examples. Note that not all the PR-channels can be written in this form, since here $x_{p,q}=y_{p,q}$. Note also that we may choose $M_\ell$ and $w_\ell$ is such a way that the channel is not classical-to-classical and neither the measurement nor the target states are separable.

\end{exm}

\begin{exm}
We next look at an example where all the parameters have the same nonzero value, namely
\begin{align*}
\dfrac{1}{3} &= z_{p} = x_{p,q} = y_{p,q},\quad \forall p,q.
\end{align*}
One can use numerical calculations to check that the corresponding matrix $B_3$ is positive. This example shows that we can have all of the parameters non-zero at the same time.
\end{exm}

\section{Conclusions} \label{sec:con}

We have show that maximal violation of the CHSH inequality requires existence of maximally incompatible two-outcome measurements and described states in $K\tmax K$ that lead to this violation. It follows that a GPT permits implementations of the PR-box if it contains a system with maximally incompatible measurements and such that the joint state space is large enough. We have applied the results to derive the implementations of PR-boxes by classical and quantum non-signaling channels. The derivation was carried out in the framework of GPTs, which opens the door for generalizations of our calculation. For the qubit case, we gave a full description of the PR-channels and proved that all such channels are necessarily entanglement-breaking.

The question of possibility of instantaneous implementation of these channels is out of the scope of this work. Obtaining some no-go theorems that would forbid such a possibility would provide further insight into Bell non-locality and our complete characterization of the qubit case might be useful for proving such results.

There is a plethora of further open questions and directions of research: one may ask about the structure of all implementations of PR-boxes for channels in higher dimensions and also for more general state spaces, one may also ask which states (and which measurements) violate the CHSH inequality more than a given number. One may also consider a resource theory of CHSH inequality violations.

Our results also raise the question of general applicability of the CHSH inequality as a test of quantumness of a system, if having too big CHSH violation constrains us to entanglement-breaking channels, which can be seen as classical channels in a sense. This also suggests the existence of some kind of trade-off between CHSH violation and some notion of quantumness of non-signaling channels.

\begin{acknowledgments}{We would like to thank the anonymous referee for giving us valuable hints and insights to improve the readability of the manuscript.} This research was supported by grant VEGA 2/0142/20 and by the grant of the Slovak Research and Development Agency under contract APVV-16-0073. MP acknowledges support from the DFG and the ERC (Consolidator Grant 683107/TempoQ).
\end{acknowledgments}

\bibliography{citations}

\appendix
\section{The Popescu-Rohrlich box}\label{sec:PRbox}
Let $x$ be a non-signaling box such that
\begin{equation*}
X_\chsh=E(\A, \B) + E(\A, \B') + E(\A', \B) - E(\A', \B') = 4.
\end{equation*}
Since all the correlations are in $[-1,1]$, we must have \[ E(\A, \B) =E(\A, \B') = E(\A', \B)= - E(\A', \B')=1. \] It is easily seen that this happens if and only if we have
\begin{align*}
P_{x}(1,1 | \A, \B) + P_{x}(-1, -1| \A, \B) &= 1, \\ P_{x}(1,1 | \A, \B') + P_{x}(-1, -1 | \A, \B') &= 1, \\ P_{x}(1,1 | \A', \B) + P_{x}(-1, -1 | \A', \B) &= 1, \\ P_{x}(1, -1 | \A', \B') + P_{x}(-1, 1 | \A', \B') &= 1,
\end{align*}
with all other probabilities equal to $0$. From this and the non-signaling conditions, we obtain
\begin{align*}
P_x (1, 1 | \A, \B) &= P_x(1,1 | \A, \B) + P_{x}(1, -1 | \A, \B) \\ &= P_{x}(1,1 | \A, \B') + P_{x}(1, -1 | \A, \B') \\ &= P_{x}(1, 1 | \A, \B').
\end{align*}
In a similar fashion one may show that all the nonzero probabilities must be equal, which implies the equality \eqref{eq:square_prob}. The case $X_\chsh=-4$ is treated similarly.

\section{Implementation of the PR-box on the square state space} \label{sec:square}
Let $S$ be the square state space and let $\A_0$, $\A_{1}$ be the maximally incompatible two-outcome measurements corresponding to the effects $\pi_0, \pi_{1} \in E(S)$ respectively. Since $\{1,\pi_0,\pi_{1}\}$ form a basis of $A(S)$, it is clear that any element $\phi\in S\tmax S\subset A(S)^*\otimes A(S)^*$ is uniquely determined by the values $(f\otimes g)(\phi)$, where $f,g\in \{\pi_0,1-\pi_0,\pi_{1},1-\pi_{1}\}$, which are exactly the outcome probabilities $P_\phi(\epsilon, \eta | \C, \D)$, $\C, \D \in \{\A_0,\A_{1}\}$, $\epsilon,\eta\in \{-1,1\}$. So if there is an implementation of the PR-box with the measurements $\{\A_0,\A_{1}\}$ on both sides, it must be unique. Moreover, it follows by Corollary \ref{coro:maxinc} that any other implementation, with other maximally incompatible pairs of measurements, is obtained by applying a local isomorphism on each copy of $S$.

Let us now find a state $\phi_S \in S \tmax S$, satisfying the equalities \eqref{eq:square_prob}. Every $\phi \in A(S)^* \otimes A(S)^*$ can be written as
\begin{equation*}
\phi = \psi_{00} \otimes s_{00} + \psi_{10} \otimes s_{10} + \psi_{01} \otimes s_{01}
\end{equation*}
for some $\psi_{00}, \psi_{10}, \psi_{01} \in A(S)^*$. By the characterization of $E(S)$ in Example \ref{exm:square}, we see by applying the maps $1 \otimes 1$, $\id \otimes \pi_0$, $\id \otimes (1-\pi_0)$, $\id \otimes \pi_{1}$, $id \otimes (1-\pi_{1})$ that $\phi \in S\tmax S$ if and only if
\begin{equation*}
1(\psi_{00} + \psi_{10} + \psi_{01}) = 1
\end{equation*}
and
\begin{align*}
&\psi_{10} \geq 0,
&&\psi_{00} + \psi_{10} \geq 0, \\
&\psi_{01} \geq 0,
&&\psi_{00} + \psi_{01} \geq 0,
\end{align*}
Writing $\psi_{00},\psi_{01},\psi_{10}$ in the basis $\{s_{00},s_{10}-s_{00},s_{01}-s_{00}\}$ and using the fact that $\{1,\pi_0,\pi_{1}\}$ is the dual basis (see Example \ref{exm:square}), we see that the required equalities hold if we put
\begin{align*}
&\psi_{00} = \dfrac{1}{2} ( s_{00} - s_{10} ),
&&\psi_{10} = \dfrac{1}{2} s_{11},
&&\psi_{01} = \dfrac{1}{2} s_{10},
\end{align*}
and it is easily checked that the conditions required for $\phi\in S\tmax S$ are satisfied (the last inequality follows from $s_{00}+s_{11}=s_{10}+s_{01}$). This gives \eqref{eq:square_PR}.

\section{Proof of Thm. \ref{thm:finding-result}} \label{sec:proof}

Since $X_{\chsh}=4$, we see that the outcome probabilities must satisfy \eqref{eq:square_prob}. Let $f_\A, f_{\A'}, f_\B,f_{\B'}\in E(K)$ be the effects corresponding to the four measurements and put
\begin{align*}
&x_{00} = 2 ((1-f_\A) \otimes \id)(\phi),
&&x_{11} = 2 (f_\A \otimes \id)(\phi) \\
&x_{01} = 2 ((1-f_{\A'}) \otimes \id)(\phi),
&&x_{10} = 2 (f_{\A'}\otimes \id)(\phi).
\end{align*}
From the fixed values of the outcome probabilities, it can be checked that $ x_{i,j}\in K$ and that for all $i, j \in \{0, 1 \}$ we have
\begin{align*}
&f_\B (x_{ij}) = i, &&f_{\B'} (x_{ij}) = j
\end{align*}
and
\begin{equation*}
x_{00} + x_{11} = x_{10} + x_{01},
\end{equation*}
i.e. the points $x_{ij}$ form a witness square for $\B, \B'$. It follows that $\B$ and $\B'$ are maximally incompatible and we can prove in the same way that $\A$ and $\A'$ are maximally incompatible as well. Let $\iota_\A, \Pi_\A$ and $\iota_\B$, $\Pi_\B$ be the affine maps from Prop. \ref{prop:GPT-condMaxInc}, resp. their positive linear extensions. Then it is easily seen from \eqref{eq:ext_pi} that we have
\begin{align*}
&\pi_0\circ \Pi_\A =f_\A,
&&\pi_{1}\circ\Pi_{\A}=f_{\A'},\\
&\pi_0\circ \Pi_\B=f_\B,
&&\pi_{1}\circ\Pi_{\B}=f_{\B'}.
\end{align*}
It follows that $(\Pi_\A\otimes \Pi_B)(\phi)\in S\tmax S$ with the pair of two-outcome measurements $\A_0$, $\A_{1}$ applied on both sides is an implementation of the PR-box, hence we must have $(\Pi_\A\otimes \Pi_B)(\phi)=\phi_S$.

Next, note that $P_\A:=\iota_\A\circ \Pi_\A$ ($P_\B:=\iota_\B\circ \Pi_\B$) is a positive projection on $A(K)^*$ onto the range of the map $\iota_\A$ ($\iota_\B$). We then have
\begin{align*}
\phi&=(P_\A\otimes P_\B)(\phi)+ (id-P_\A\otimes P_\B)(\phi)\\
&=(\iota_\A\otimes \iota_\B)(\phi_S)+\phi^\perp,
\end{align*}
where $(P_\A\otimes P_\B)(\phi^\perp)=0$, which implies that \[ (\Pi_\A\otimes\Pi_\B)(\phi^\perp)=(\Pi_\A\otimes\Pi_\B)\circ (P_\A\otimes P_\B)(\phi^\perp)=0. \] Now note that using \eqref{eq:ext_pi} and the definition of $x_{ij}$, we get
\begin{align*}
(\Pi_\A\otimes id)(\phi)&=\frac12((s_{00}-s_{10})\otimes x_{00} \\
&+ s_{11}\otimes x_{10}+ s_{10}\otimes x_{01})\\
&= (id\otimes \iota_\B)(\phi_S)=(\Pi_\A\otimes id)((\iota_\A\otimes\iota_\B)(\phi_S))
\end{align*}
and hence $(\Pi_\A\otimes id)(\phi^\perp)=0$. In a similar manner, we obtain that $(id\otimes \Pi_\B)(\phi^\perp)=0$, which implies that $\phi^\perp\in \ker(\Pi_\A)\otimes\ker(\Pi_\B)$.

\begin{table*}
\caption{All possible forms of $C$. \label{tbl:PRproof-C}}
\begin{ruledtabular}
\begin{tabular}{cccc}
$U_A/U_B$ & $\neq$ & $\I$ & $V$ \\
\hline
\multirow{3}{*}{$\neq$} & $r=1$ & $r=2$ & $r=2$ \\
& & $\Delta_1=\{00,01\}$ & $\Delta_1=\{10,11\}$ \\
& & $V_{\out}=W$ & $V_{\out}=(U_A\otimes \I) W$ \\
\hline
\multirow{3}{*}{$\I$} & $r=2$ & $r=3$ & $r=3$ \\
& $\Delta_1=\{01,11\}$ & $\Delta_0=\{11\}$ & $\Delta_0=\{01\}$ \\
& $V_{\out}=W$ & $V_{\out}=W$ & $V_{\out}=W$ \\
\hline
\multirow{3}{*}{$V$} & $r=2$ & $r=3$ & $r=3$ \\
& $\Delta_1= \{00,10\}$ & $\Delta_0=\{10\}$ & $\Delta_0=\{00\}$ \\
& $V_{\out}=(\I\otimes U_B)W$ & $V_{\out}=W$ & $V_{\out}=(\I\otimes V)W$
\end{tabular}
\end{ruledtabular}
\end{table*}

\section{The proof of Prop. \ref{prop:PRchan}}\label{sec:PRproof}
Let $\A$, $\A'$ be the pair of maximally incompatible two-outcome measurements corresponding to the choice of the ONB $|\eta^A_0\>$, $|\eta^A_1\>$ (Prop. \ref{prop:max_inc}) and similarly let $\B$, $\B'$ be the measurements for $|\eta^B_0\>$, $|\eta_1^B\>$. Let $U_A,U_B:\Ha\to \Ha$ denote the unitaries given as \[ U_A|i\>= |\eta^A_i\>,\ U_B|i\>= |\eta^B_i\>. \]

Let $\Phi$ be a PR-channel that maximally violates the CHSH inequality with this choice of measurements and let $C=C(\Phi)$ have the block-diagonal form \[ C=\sum_{\alpha,\beta\in \{0,1\}^2} C_{\alpha,\beta}\otimes |\alpha\>\<\beta|_{\inn}. \] Then $\Phi$ must satisfy \eqref{eq:square_prob}, where for all $\C\in \{\A,\A'\}$ and $\D\in \{\B,\B'\}$, we have
\begin{align*}
P_\Phi(1,1 | \C,\D)&= \Tr( \swap_{23}(F_\C\otimes F_\D)C)\\ P_\Phi(-1, 1 | \C,\D)&= \Tr( \swap_{23}(F^\perp_\C\otimes F_\D)C)\\ P_\Phi(1, -1 | \C,\D)&= \Tr( \swap_{23}(F_\C\otimes F^\perp_\D)C)\\ P_\Phi(-1, -1 | \C,\D)&= \Tr( \swap_{23}(F^\perp_\C\otimes F^\perp_\D)C),
\end{align*}
here $F_\C, F_C^\perp$ are the PPOVM operators corresponding to $\C\in \{\A,\A'\}$ and $F_\D,F_\D^\perp$ correspond to $\D\in \{\B,\B'\}$. Let $Q_h$, $h=1,\dots,16$ be operators obtained as \[ Q_h=\swap_{23}(H_A\otimes H_B), \] where $H_\A\in \{F_\A,F_\A^\perp,F_{\A'},F_{\A'}^\perp\}$ and similarly for $H_B$. Then $Q_h$ are mutually orthogonal rank 1 product projections and all the values of $\Tr(Q_hC)$ are either 0 or $1/2$. Let \[ P=\sum \{Q_h,\ \Tr(Q_hC)=1/2\}. \] Since $C\ge 0$ and $\Tr((\I-P)C)=0$, we must have $C=PCP$. One can check by \eqref{eq:square_prob} that \[ P=\sum_{\alpha\in \{0,1\}^2} U_\alpha U^*_\alpha\otimes |\alpha\>\<\alpha|_{\inn} \] for the isometries $U_\alpha: \mathbb C^2\to \Ha_{\out}$, given as \[ U_{km}=(U_A^k\otimes U_B^mV^{m.k})W,\quad k,m=0,1, \] where $V=
\begin{pmatrix}
0 & 1 \\ 1 & 0
\end{pmatrix}
$ and $W: \mathbb C^2\to \Ha_{\out}$ is the isometry given as $W|i\>=|ii\>$. It follows that \[ C=PCP=\sum_{\alpha,\beta} U_\alpha U_\alpha^*C_{\alpha,\beta}U_\beta U_\beta^*\otimes |\alpha\>\<\beta|_{\inn}, \] so that for $\alpha=km$, $\beta=ln$ we have \[ C_{\alpha,\beta}=\sum_{i,j} c^{\alpha,\beta}_{ij}U_A^k|i\>\<j|U_A^{-l}\otimes U_B^mV^{k.m}|i\>\<j|V^{-l.n}U_B^{-n} \] for some coefficients $c^{\alpha,\beta}_{ij}$.

Since $\Phi$ is non-signaling, the Choi matrix must satisfy the conditions
\begin{align*}
&\Tr_{A,out} (C) = \I_{A,in} \otimes C_B,
&&\Tr_{B,out} (C) = \I_{B,in} \otimes C_A,
\end{align*}
where $C_B\in B^+(\Ha_{B,out}\otimes \Ha_{B,in})$ and $C_A\in B^+(\Ha_{A,out}\otimes \Ha_{A,in})$. This amounts to
\begin{align}
&\Tr_A (C_{km,ln}) = 0,
&&k \neq l,
&&\forall m,n
\label{eq:partialA_neq} \\
&\Tr_A (C_{0m,0n}) = \Tr_A C_{1m,1n},
&&
&&\forall m,n
\label{eq:partialA_eq} \\
&\Tr_B(C_{km,ln}) = 0,
&&m \neq n,
&&\forall k,l
\label{eq:partialB_neq} \\
&\Tr_B(C_{k0,l0}) = \Tr_B(C_{k1,l1}),
&&
&&\forall k,l.\label{eq:partialB_eq}
\end{align}

If $k=l$, we have
\begin{equation*}
\Tr_A (C_{km,kn}) = \sum_i c^{km,kn}_{ii} U_B^mV^{k.m}|i\>\<i|V^{-k.n}U_B^{-n}.
\end{equation*}
By \eqref{eq:partialA_eq}, we see that we must have $c_{ii}^{km,kn}=0$, $\forall i$ whenever $m\ne n$. Similarly, by \eqref{eq:partialB_eq} we obtain for $k\ne l$ that $c_{ii}^{k0,l0}=0$, $\forall i$ and $c_{ij}^{k1,l1}=0$, $\forall i\ne j$. Note also that $c_{ii}^{km,km}=1/2$ by \eqref{eq:square_prob}.

Looking at the conditions \eqref{eq:partialA_neq} and \eqref{eq:partialB_neq}, we see that there are two different possibilities on each side: either $\<i|U_{A}|j\>=0$ for some $i,j$ or $\<i|U_{A}|j\>\ne 0$ for all $i,j$, the same for $U_B$. The first condition means that $U_A$ is diagonal or off-diagonal (i.e. $\<i| U_A|i\>=0$). Since the diagonal (off-diagonal) elements in these cases only correspond to scalar factors of the basis vectors, they may safely be put to 1, so that $U_A=\I$ ($U_A=V$).

Assume that $U_A\ne \I$, $U_A\ne V$, then we obtain from \eqref{eq:partialA_neq} that $C_{km,ln}=0$ whenever $k\ne l$. There are essentially two off-diagonal blocks left:
\begin{equation*}
C_{00,01}=C_{01,00}^*=\sum_{i\ne j} c^{00,01}_{ij} |i\>\<j|\otimes |i\>\<j|U_B^*
\end{equation*}
and
\begin{equation*}
C_{10,11}(=C_{11,10}^*)=\sum_{i\ne j} c^{11,01}_{ij} U_A|i\>\<j|U^*_A\otimes |i\>\<j|V^*U^*_B.
\end{equation*}
The condition \eqref{eq:partialB_neq} implies that $C_{00,01}\ne 0$ only if $U_B$ is diagonal, in which case $U_B=\I$ and
\begin{equation*}
C_{00,01} = \sum_{i\ne j} c^{00,01}_{ij} |i\>\<j|\otimes |i\>\<j|.
\end{equation*}
Similarly, we can have $C_{11,01}\ne 0$ only if $U_B=V$, in which case
\begin{equation*}
C_{11,01} = \sum_{i\ne j} c^{11,01}_{ij} U_A|i\>\<j|U_A^*\otimes |i\>\<j|.
\end{equation*}
It follows that if also $U_B\ne \I$, $U_B\ne V$, then $C=\sum_\alpha C_{\alpha,\alpha}\otimes |\alpha\>\<\alpha|$ is block-diagonal, which means that is has the form \eqref{eq:PRform} with $r=1$. If $U_B=\I$ then we have $\Delta_1=\{00,01\}$ and $V_{\out}=U_{00}=U_{01}=W$, for $U_B=V$ we obtain $\Delta_1=\{10,11\}$ and $V_{\out}=U_{10}=U_{11}=(U_A\otimes\I)W$. We may apply similar reasoning in all cases with $U_B\ne \I$, $U_B\ne V$.

Next we turn to the cases when $U_A,U_B\in \{\I,V\}$. We will provide a proof for $U_A=U_B=\I$, all other cases are similar. Here we obtain from \eqref{eq:partialA_neq} - \eqref{eq:partialB_eq} that all off-diagonal blocks $C_{\alpha,\beta}$ must have $c_{ii}^{\alpha,\beta}=0$ and $C_{km,ln}=0$ whenever $k.m\ne l.n$, which implies that $C_{\alpha,11}=0$ for all $\alpha\ne 11$. Note that $U_{00}=U_{10}=U_{01}=W$, so that $C$ has the form \eqref{eq:PRform} with $\Delta_0=\{11\}$ and $V_{\out}=W$. All possible forms of $C$ in the different cases are summarized in table \ref{tbl:PRproof-C}.

\end{document}